\documentclass[12pt,draftclsnofoot,onecolumn]{IEEEtran}

\usepackage{amsfonts}
\usepackage{times}
\usepackage{latexsym}
\usepackage{amssymb}
\usepackage{amsmath}
\usepackage{cite}
\usepackage{verbatim}

\newcommand{\bydef}{\triangleq}

\def\bydef{:=}

\def\bb0{{\mathbb{0}}}

\def\bydef{:=}

\def\bb{{\mathbf{b}}}

\def\bh{{\mathbf{h}}}

\def\bq{{\mathbf{q}}}

\def\bw{{\mathbf{w}}}
\def\bx{{\mathbf{x}}}
\def\by{{\mathbf{y}}}

\def\b0{{\mathbf{0}}}

\def\bA{{\mathbf{A}}}

\def\bH{{\mathbf{H}}}

\def\bbC{{\mathbb{C}}}

\def\bbE{{\mathbb{E}}}

\def\bydef{:=}

\def\sf0{{\mathsf{0}}}

\usepackage{graphicx}
\usepackage{amssymb}
\usepackage{amsfonts}
\usepackage{amsmath}
\usepackage{latexsym}
\begin{document}

\newtheorem{thm}{Theorem}
\newtheorem{lemma}{Lemma}
\newtheorem{rem}{Remark}
\newtheorem{exm}{Example}
\newtheorem{prop}{Proposition}
\newtheorem{defn}{Definition}
\newtheorem{cor}{Corollary}
\def\proof{\noindent\hspace{0em}{\itshape Proof: }}
\def\endproof{\hspace*{\fill}~\QED\par\endtrivlist\unskip}
\def\bh{{\mathbf{h}}}
\def\SIR{{\mathsf{SIR}}}
\title{How Many Transmit Antennas to Use in a MIMO Interference Channel}
\author{
Rahul~Vaze
\thanks{Rahul~Vaze is with the School of Technology and Computer Science, Tata Institute of Fundamental Research, Homi Bhabha Road, Mumbai 400005, vaze@tcs.tifr.res.in. }}

\date{}
\maketitle
\noindent
\begin{abstract} 
The problem of finding the optimal number of data streams to transmit in a multi-user MIMO scenario, where both the transmitters and receivers are equipped with multiple antennas is considered. 
Without  channel state information  at any transmitter, with a zero-forcing receiver each user is shown to transmit a single data stream to maximize its own outage capacity in the presence of sufficient number of users. Transmitting a single data stream is also shown to  be optimal in terms of maximizing the sum of the outage capacities in the presence of sufficient number of users. 
 
\end{abstract}

\section{Introduction}
Employing multiple antennas at transmitters and receivers is well known to improve the performance of wireless communication by either decreasing the bit-error rate (BER) \cite{Tarokh1999a}, or increasing the channel capacity \cite{Telatar1999,Foschini1998}. 
A key assumption used in  \cite{Tarokh1999a,Telatar1999,Foschini1998} is that the transmission is interference free, i.e. each multiple antenna equipped receiver only receives signal from its corresponding multi-antenna transmitter, and no other transmitter is transmitting at the same time. This assumption is easy to justify in practice with base-station based centralized controllers, however, it fails to work in decentralized 
wireless network setting such as sensor networks, ad-hoc networks etc., where there are large number of uncoordinated  transmitters. 
In a decentralized wireless network, each node transmits independently, thereby possibly causing interference to all the other nodes. 

We  consider a decentralized wireless network setting, where there are $N$ non-cooperating transmitter-receiver pairs or links,  and all the transmitter and receiver nodes are  equipped with $M$ antennas (popularly known as the MIMO interference channel). 
We assume that no transmitter has any channel state information (CSI), while each receiver has CSI only for the channel from its corresponding 
transmitter. With its CSI, we assume that  each receiver uses a zero-forcing ZF receiver to decode the different streams sent from its corresponding transmitter. The ZF receiver is considered because of its low complexity implementation. Our results can be generalized to MMSE receivers as well. With ZF, we consider outage capacity as the performance metric owing to its analytical tractability and practical significance.
Outage capacity is defined as the rate of transmission multiplied with the probability that the transmission is not in outage \cite{Telatar1999}, where outage is defined as the event that the mutual information of the channel is less than the target rate of transmission. In this paper we are interested in finding the optimal number of data streams to send from each transmitter that maximizes the sum outage capacity of the MIMO interference channel with ZF. 

%

\subsection{Contributions}
The contributions of this paper are as follows
\begin{itemize}
\item We show that with sufficiently large number of users $N$, \footnote{Of the order of the number of antennas $M$.} to maximize the individual outage capacity it is optimal for each user to transmit a single data stream irrespective of the  number of data streams transmitted by other users. 

\item  For sufficiently large  number of users $N$, we show that the sum outage capacity is maximized when each user transmits a single data stream.

\end{itemize}

\subsection{Comparison with prior work}

When no transmitter has any CSI, under an average power constraint, Shannon capacity is equal to the expected mutual information that is obtained using the maximum likelihood (ML) decoding. ML decoding, however, is quite complicated in a multi-user MIMO scenario, and finding the optimal number of data streams to transmit that maximize the mutual information with ML decoding is quite challenging. \footnote{ In prior work,  \cite{Blum2003} studied the expected mutual information as a function of data streams  for extremely large interference power.}
Thus, for analytical tractability and to get insights into the problem, we consider a simple ZF decoder \cite{Booktse}, 
and use outage capacity as the performance metric. Outage framework has been  extensively used in past to understand the performance of multiple antenna systems \cite{Telatar1999, Zheng2003,Vaze2009}. 
In this paper we show that transmitting a single data stream is selfishly optimal for each user to maximize its own outage capacity in the presence of sufficient number of transmitter-receiver pairs.

Without any CSI at the transmitter, finding the optimal number of data streams that maximize the sum capacity \footnote{Though different capacity definitions have been used in literature.} has attracted a lot of attention \cite{Blum2003,Mackay2009,Vaze2009}. For a large ad-hoc network, where the transmitter locations are distributed as a Poisson point process,  
the optimal number of data streams to transmit that maximize the transmission capacity \cite{Weber2005} has been derived in \cite{Vaze2009, Mackay2009} for $M, N \rightarrow \infty$.
When each receiver employs interference cancelation, single data stream transmission has been shown to be optimal  \cite{Vaze2009}, while without interference cancelation, using the number of data streams equal to a fraction
of the total transmit antennas \cite{Vaze2009, Mackay2009} has been shown to maximize the transmission capacity. 
Single data stream transmission has also been shown to maximize the sum of the  ergodic Shannon capacities  \cite{Blum2003} for the limiting case of extremely large interference power. Unlike previous works \cite{Mackay2009,Vaze2009}, we consider the finite $N,M$ regime, and do not make any assumptions on the location of transmitters/receivers or the interference power at any receiver in contrast to \cite{Blum2003}. 
Without any of the mentioned assumptions, for an arbitrary MIMO interference channel, in this paper we  
show that if the number of links $N$ is large enough (typically $\approx M$ when the SIR threshold required for correct decoding is greater than $1$, where $M$ is the number of antennas at each node), then transmitting a single data stream from each user is optimal for maximizing the 
sum outage capacity.

{\it Notation:}
Let ${\bA}$ denote a matrix, ${\bf a}$ a vector and
$a_i$ the $i^{th}$ element of ${\bf a}$. Transpose and conjugate transpose is denoted by $^T$, and $^*$, respectively. 
A circularly symmetric complex Gaussian random
variable $x$ with zero mean and variance $\sigma^2$ is denoted as $x
\sim {\cal CN}(0,\sigma^2)$. A chi-square distributed random variable $y$ with $m$ degrees of freedom is denoted by 
$y \sim \chi^2(m)$.  t t We use the symbol
$\bydef$  to define a variable.

\section{System Model} 
Consider a MIMO interference channel with $N$ transmitter-receiver pairs, where each transmitter and receiver is equipped with $M$ antennas. We assume that each  transmitter is only interested in transmitting to its corresponding receiver and has an average power constraint of $P$. 
We assume that each receiver knows CSI for its corresponding transmitter, however,  no CSI is available at any transmitter.
\footnote{ In this paper we do not consider the availability of CSI at each transmitter, since solving that case requires simple closed form expression for the PDF of all the eigenvalues of channel matrices between transmitters and receivers, which unfortunately is not available.} 
With no CSI, the $n^{th}$ transmitter sends data vector $\bx_n\in \bbC^{k_n\times 1}$ consisting of $k_n$ data streams, where each data stream  is independent and ${\cal CN}(0,1)$ distributed, using its $k_n$ antennas by distributing its power uniformly over the $k_n$ antennas.
With this model the received signal at the $n^{th}$ receiver is 
\begin{equation}
\label{eq:rxsig}
\by_n = \sqrt{\frac{P}{k_n}} \bH_{nn}\bx_n  + \sum_{m=1, m\ne n}^N \sqrt{\frac{P}{k_m}}\bH_{mn}\bx_m + \bw_n,
\end{equation}
where   $\bH_{mn} \in \bbC^{N\times k_n}$ is the channel coefficient matrix between the $m^{th}$ transmitter and the $n^{th}$ receiver whose entries are  i.i.d.  ${\cal C N}(0,1)$, and $\bw_n$ is the additive white Gaussian noise with zero mean and $\sigma^2$  variance.
For sufficiently large $N, P,$ this system is interference limited and we drop the AWGN contribution in the sequel.

We assume that each receiver decodes the $k_n$ data streams independently using a ZF decoder \cite{Booktse}. Hence to decode the 
$j^{th}$ stream out of the total $k_n$ streams at the $n^{th}$ receiver, the received signal is projected onto the null space of the channel coefficient vectors corresponding to the $[1,2, \ldots, j-1, j+1, \ldots, k_n]$ data streams. Thus the $n^{th}$ receiver multiplies $\bq^n_j$ 
to the received signal $\by_n$ to decode its $j^{th}$ stream, if $\bq_j^n \in {\cal N}\left([\bH_{nn}(1)\ldots  \bH_{nn}(j-1) \bH_{nn}(j+1)\ldots  \bH_{nn}(k_n)]\right)$, where ${\cal N}\left([P]\right)$ represents the null space of columns of $P$, and $\bH_{nn}(\ell)$ represents the $\ell^{th}$ column of $\bH_{nn}$. 

From (\ref{eq:rxsig}), using the ZF decoder, the signal-to-interference ratio (SIR) for the $j^{th}$ stream is 
  \begin{equation}
\label{eq:sir}
\SIR_j^n = \frac{\frac{P}{k_n}|\bq_{j}^n\bH_{nn}(j)|^2}{ \sum_{m=1, m\ne n}^N \frac{P}{k_m} \sum_{\ell=1}^{k_m} |\bq_{j}^n\bH_{mn}(\ell)|^2}.
\end{equation}
Note that $\SIR_j^n$  is identically distributed for $j=1,2,\dots, k_n$ for a fixed $n$, $n=1,2,\ldots, N$. To simplify the notation  let $s_{j}^n \bydef |\bq_{j}^n\bH_{nn}(j)|^2$, and $I_{\ell, m}\bydef  |\bq_{j}^n\bH_{mn}(\ell)|^2$. From  \cite{Jindal2008a}, $s_{j}^n \sim \chi_{2(M-k+1)}^2$ and $I_{\ell,m}\sim  
\chi_{2}^2, \ \forall \ j, n,\ell,m$. Hence 
$\SIR_j^n = \frac{\frac{P}{k_n}s_j^n}{ \sum_{m=1, m\ne n}^N \frac{P}{k_m} \sum_{\ell=1}^{k_m} I_{\ell,m}}$.

We assume that a fixed rate of $R$ bits/sec/Hz is transmitted on each data stream, and transmission on any data stream is deemed to be  successful  if the SIR on that data stream is larger than a threshold $\beta$, which is a function of $R$, i.e., the transmission is not in outage. Hence the successful rate (outage capacity) obtained on any data stream is the product of $R$ and the probability that the SIR on that link is larger than $\beta$, $P(\SIR_{j}^n \ge \beta)$. Combining all the $k_n$ streams the outage capacity on the $n^{th}$ link is $C_n \bydef k_nR P(\SIR_{j}^n \ge \beta)$ bits/sec/Hz. Hence the sum outage capacity of the MIMO interference channel with ZF  is $C_{sum} \bydef \sum_{n=1}^NC_n$, and we want to find the optimal $k_1,\dots, k_n$ that maximize the sum outage capacity, i.e. $\arg \max_{k_1,\dots,k_n} C_{sum}$. 

We first solve $\arg \max_{k_j} C_{j}, \forall \ j$, and use that to solve  $\arg \max_{k_1,\dots,k_n} C_{sum}$ as follows.
%

\begin{thm}\label{thm:NE} Using ZF decoder at each receiver, for any choice of $k_1, \dots, k_{j-1}, k_{j+1}, \dots, k_n$, $k_j=1$ (transmitting a single data stream from the $j^{th}$ user) maximizes the outage capacity $C_j$ of the $j^{th}$ user  for sufficiently large $N$ (specified in the proof).
\end{thm}
\begin{proof} See Appendix \ref{app:NE}.
\end{proof}
\begin{rem}\label{rem:boundN}  In Appendix \ref{app:NE} we show that for $\beta > 1$, Theorem \ref{thm:NE} holds for $N  \ge \theta M$, where $\theta >1$ is a constant. Simulation results indicate that $N \approx M$ is sufficient for Theorem \ref{thm:NE} to hold when $\beta > 1$. Typically $\beta = 2^R-1$, where $R$ is the rate of transmission in bits/sec/Hz. Thus, for $R>1$ bits/sec/Hz, $N \approx M$ is sufficient for Theorem \ref{thm:NE} to hold. Moreover, in any practical system the number of users is much larger than the number of antennas at each node, hence Theorem \ref{thm:NE} is applicable for most practical scenarios.
\end{rem}

%
Next, we use Theorem \ref{thm:NE} to solve the general problem $\arg \max_{k_1,\dots,k_n} C_{sum}$.
\begin{cor}  Using ZF decoder at each receiver, $k_n=1 \ \forall \ n$ maximizes the sum capacity $C$, i.e. $(1,1,\ldots,1) = \arg \max_{k_1, \ldots k_N} C$ for sufficiently large $N$ (specified in Theorem \ref{thm:NE}).
\end{cor}
\begin{proof} Recall that $C = \sum_{n=1}^NC_n =  \sum_{n=1}^N  R k_n P(\SIR_{j}^n \ge \beta)$. From Theorem \ref{thm:NE}, $k_n=1$ maximizes 
$C_n$ for any value of $k_1, \ldots, k_{n-1}, k_{n+1}, \ldots, k_{N}$ for sufficiently large $N$. Thus, 
\begin{eqnarray*}
C_n &=& R k_n P(\SIR_{j}^n \ge \beta),\\ 
&\le& R  P(\SIR_{j}^n \ge \beta),  \ \ \ \ \text{from Theorem \ref{thm:NE}},\\
 &=& RP\left(\frac{s_j^n }{ \sum_{m=1, m\ne n}^N \frac{1}{k_m} \sum_{\ell=1}^{k_m} I_{\ell, m}} \ge \beta \right),  \ \ \ \text{from the optimality of} \ k_n=1, 
\end{eqnarray*}
where $s_j^n \sim \chi^2_{2(M)}$ and $I_{\ell,m} \sim \chi^2_{2}, \ \forall \ \ell, m$.
Thus, $C \le R \sum_{n=1}^N P\left(\frac{s_j^n }{ \sum_{m=1, m\ne n}^N \frac{1}{k_m} \sum_{\ell=1}^{k_m} I_{\ell, m}} \ge \beta \right)$. Moreover for any $n$, since $P\left(\frac{s_j^n }{ \sum_{m=1, m\ne n}^N \frac{1}{k_m} \sum_{\ell=1}^{k_m} I_{\ell, m}} \ge \beta \right)$ is a decreasing function of $k_1, \ldots, k_{n-1},$ $k_{n+1}, \ldots, k_N$, $P\left(\frac{s_j^n }{ \sum_{m=1, m\ne n}^N \frac{1}{k_m} \sum_{\ell=1}^{k_m} I_{\ell, m}} \ge \beta \right)$ is maximized at $k_m=1, \ \forall \ m =1,\dots,n-1, n+1,\ldots,N$ for each $n$. Hence 
$C \le R \sum_{n=1}^N P\left(\frac{s_j^n }{ \sum_{m=1, m\ne n}^N I_{ m}} \ge \beta \right)$, $I_{m} \sim \chi^2_{2}, \ \forall \  m$. Clearly, using $k_n=1, \ \forall \ n =1,2,\ldots,N$ we can achieve this upper bound, which concludes the proof.
\end{proof}

{\it Discussion:} 
In this section we derived that transmitting a single data stream maximizes the outage capacity of each link in the presence of sufficient number of links. An intuitive justification of this result is that with a sufficient number of interfering links, the decrease in the outage probability with increasing the number of data streams outweighs the linear increase in the outage capacity by sending multiple links. Even though our result is valid for sufficiently high number of links, however, as pointed in Remark \ref{rem:boundN}, for reasonable values of threshold $\beta$,  the required number of links for our result to hold is of the order of the number of antennas which is true in most practical applications.

An important byproduct of our analysis is that it allows us to  derive the optimal number of data streams 
to send from each transmitter that maximizes the sum of the individual outage capacities.
Directly finding the optimal number of data streams to send from each transmitter that maximize the sum of the individual outage capacities is a challenging problem.  Using a two step approach, first we show that  transmitting a single data stream selfishly maximizes the outage capacity of each user.  Then using the fact that the outage capacity of any user is a decreasing function of the number of data streams used by other users, we conclude  that transmitting a single data stream is globally optimal to maximize the sum of the individual outage capacities.

\section{Simulations}
In this section we provide some numerical examples to illustrate the results obtained in this paper.  In Fig. \ref{fig:M=10} we plot the outage capacity of any one user (say the 1st user) versus the number of data streams $k_1$ it uses with $M=10, \beta =1$, when all other users (interferers for 1st user) use a single data stream $k_n=1, \ n\ne 1$ for several values of $N$. We see that as $N$ goes towards $M$, 
$k_1=1$ becomes optimal for maximizing the individual outage capacity. Thus, for $\beta \approx 1$, we can see that if $N\approx M$, then $k_n=1$ maximizes the individual outage capacity in this case.
Next, in  Fig. \ref{fig:M=5} we plot the outage capacity of  any one user (say the 1st user) versus the number of data streams $k_1$ with $M=5, N=5$, when all other users (interferers for 1st user) use a single data stream $k_n=1, \ n\ne 1$ for several values of $\beta$. We can see from Fig. \ref{fig:M=5} that as $\beta$ increases, the value of $N$ required for having $k_1=1$ optimal in terms of maximizing the individual outage capacity decreases. In Fig. \ref{fig:M=3N=3}, we use $N=3, M=3$, i.e. $3$ users 
with $3$ antennas each, and  plot the sum outage capacity as function of number of data streams sent by each user $k_1,k_2,k_3$. 
From Fig. \ref{fig:M=3N=3} it follows that $k_1=k_2=k_3=1$ maximizes the sum outage capacity for $\beta=1$. Here again for $N\approx M$, it is optimal to use $k_n=1$.

\appendices

\section{Proof of Theorem \ref{thm:NE}.}
\label{app:NE}
Recall that $C_n =  R k_n P(\SIR_{j}^n \ge \beta)$, where $P\left(\frac{\frac{s_j^n}{k_n} }{ \sum_{m=1, m\ne n}^N \frac{1}{k_m} \sum_{\ell=1}^{k_m} I_{\ell, m}} \ge \beta \right)$, and $s_j^n \sim \chi^2_{2(M-(k_n-1))}$, and $I_{\ell,m} \sim \chi^2_{2}, \ \forall \ \ell, m$. Hence 
\begin{eqnarray}\nonumber
C_n &=&  R k_n P\left(\frac{\frac{s_j^n}{k_n} }{ \sum_{m=1, m\ne n}^N \frac{1}{k_m} \sum_{\ell=1}^{k_m} I_{\ell, m}} \ge \beta \right), 
\\\label{eq:appexpC}
&=& R k_n \bbE_{I}\left\{\sum_{r=1}^{M-k_n+1}\frac{(\beta k_n I)^r}{r!} e^{-\beta k_n I}\right\}, \ \text{since}\  s_j^n \sim \chi^2_{2(M-(k_n-1))},
\end{eqnarray}
where $I \bydef \sum_{m=1, m\ne n}^N \frac{1}{k_m} \sum_{\ell=1}^{k_m} I_{\ell, m}$.

{\bf Case 1: $k_m=k, \forall \  m \ne n$}

In this case $kI \sim \chi^2_{2((N-1)k)}$, and hence 
\begin{eqnarray*}
C_n &=& R k_n \bbE_{I}\left\{ \sum_{r=1}^{M-k_n+1}\frac{(\frac{\beta k_n}{k} kI)^r}{r!} e^{-\frac{\beta k_n}{k} kI}\right\},\\
&=&R k_n \int_{0}^{\infty}  \sum_{r=1}^{M-k_n+1}\frac{(\frac{\beta k_n}{k} x)^r}{r!} e^{-\frac{\beta k_n}{k} x} \frac{x^{(N-1)k-1}}{((N-1)k-1 )!} e^{-x} dx \\
&=& R k_n \sum_{r=1}^{M-k_n+1}  \frac{(\frac{\beta k_n}{k})^r}{r!}\int_{0}^{\infty} \frac{x^{r+ (N-1)k-1}}{((N-1)k-1 )!} e^{-x(1+\frac{\beta k_n}{k})} dx, \\
&=& R k_n \sum_{r=1}^{M-k_n+1}  \frac{(\frac{\beta k_n}{k})^r}{(1+\frac{\beta k_n}{k})^{r+(N-1)k-1}}  \frac{(r+ (N-1)k-1)!}{r! ((N-1)k-1)!}.
\end{eqnarray*}
Let $B_n(k_n) \bydef \sum_{r=1}^{M-k_n+1}  \frac{(\frac{\beta k_n}{k})^r}{(1+\frac{\beta k_n}{k})^{r+(N-1)k-1}}  \frac{(r+ (N-1)k-1)!}{r! ((N-1)k-1)!}$. Hence $C_n = R k_n B_n(k_n)$.
To show that $C_n$ is maximized at $k_n=1$, we show that $\frac{B_n(k_n=p)}{B_n(k_n=p+1)} \ge \frac{p+1}{p}$ for $p=1,2,\ldots,M$ for large enough $N$. Towards that end, note that 
\begin{equation} \label{eq:appaux1}
\left(\frac{(\frac{\beta p}{k})}{(1+\frac{\beta p}{k})}\right) \ge \left(\frac{(\frac{\beta }{k})}{(1+\frac{\beta }{k})}\right) \ \text{for}  \ p=1,2,\ldots,M.
\end{equation}
 Similarly, 
\begin{equation}\label{eq:appaux2}
\left(\frac{(\frac{\beta p}{k})}{(1+\frac{\beta p}{k})}\right)^r \le 1,  \ \text{for} \  p=1,2,\ldots,M.
\end{equation}
Now consider 
\begin{eqnarray*}
 \frac{B_n(k_n=p)}{B_n(k_n=p+1)}  & = & \frac{\sum_{r=1}^{M-p+1}  \frac{(\frac{\beta p}{k})^r}{(1+\frac{\beta p}{k})^{r+(N-1)k-1}}  \frac{(r+ (N-1)k-1)!}{r! ((N-1)k-1)!} } {\sum_{r=1}^{M-(p+1)+1}  \frac{(\frac{\beta (p+1)}{k})^r}{(1+\frac{\beta (p+1)}{k})^{r+(N-1)k-1}}  \frac{(r+ (N-1)k-1)!}{r! ((N-1)k-1)!}}, \\
& = & \frac{\frac{1}{(1+\frac{\beta (p)}{k})^{(N-1)k-1}}\sum_{r=1}^{M-p+1} \frac{(\frac{\beta p}{k})^r}{(1+\frac{\beta p}{k})^{r}}  \frac{(r+ (N-1)k-1)!}{r! ((N-1)k-1)!} }{\frac{1}{(1+\frac{\beta (p+1)}{k})^{(N-1)k-1}}\sum_{r=1}^{M-p} \frac{(\frac{\beta (p+1)}{k})^r}{(1+\frac{\beta (p+1)}{k})^{r}}  \frac{(r+ (N-1)k-1)!}{r! ((N-1)k-1)!}},\\
&\ge & \frac{\frac{1}{(1+\frac{\beta p}{k})^{(N-1)k-1}}\sum_{r=1}^{M-p+1} \frac{(\frac{\beta }{k})^r}{(1+\frac{\beta }{k})^{r}}  \frac{(r+ (N-1)k-1)!}{r! ((N-1)k-1)!} }{\frac{1}{(1+\frac{\beta (p+1)}{k})^{(N-1)k-1}}\sum_{r=1}^{M-p}   \frac{(r+ (N-1)k-1)!}{r! ((N-1)k-1)!}}, \ \text{from}  \ \ (\ref{eq:appaux1}), (\ref{eq:appaux2})\\
&\ge & \frac{\frac{1}{(1+\frac{\beta p}{k})^{(N-1)k-1}}\sum_{r=1}^{M-p+1} \frac{(\frac{\beta }{k})^{M-p+1}}{(1+\frac{\beta }{k})^{M-p+1}}  \frac{(r+ (N-1)k-1)!}{r! ((N-1)k-1)!} }{\frac{1}{(1+\frac{\beta (p+1)}{k})^{(N-1)k-1}}\sum_{r=1}^{M-p}   \frac{(r+ (N-1)k-1)!}{r! ((N-1)k-1)!}}, \ \text{since}  \ \ \left(\frac{(\frac{\beta }{k})}{(1+\frac{\beta }{k})}\right) \le 1, \\
&= & \frac{\frac{1}{(1+\frac{\beta p}{k})^{(N-1)k-1}} \frac{(\frac{\beta }{k})^{M-p+1}}{(1+\frac{\beta }{k})^{M-p+1}}\sum_{r=1}^{M-p+1}   \frac{(r+ (N-1)k-1)!}{r! ((N-1)k-1)!} }{\frac{1}{(1+\frac{\beta (p+1)}{k})^{(N-1)k-1}}\sum_{r=1}^{M-p}   \frac{(r+ (N-1)k-1)!}{r! ((N-1)k-1)!}},\\
&\ge& \frac{\frac{1}{(1+\frac{\beta p}{k})^{(N-1)k-1}} \frac{(\frac{\beta }{k})^{M-p+1}}{(1+\frac{\beta }{k})^{M-p+1}}}{\frac{1}{(1+\frac{\beta (p+1)}{k})^{(N-1)k-1}}},\\
&=& \left(\frac{k+\beta (p+1)}{k+\beta p}\right)^{(N-1)k-1} \frac{(\frac{\beta }{k})^{M-p+1}}{(1+\frac{\beta }{k})^{M-p+1}}.
\end{eqnarray*}
Now since $\left(\frac{k+\beta (p+1)}{k+\beta p}\right) >1 $, 
and  $\frac{(\frac{\beta }{k})^{M-p+1}}{(1+\frac{\beta }{k})^{M-p+1}}$ is independent of $N$, there exists an $N$ for which
$\left(\frac{k+\beta (p+1)}{k+\beta p}\right)^{(N-1)k-1} \frac{(\frac{\beta }{k})^{M-p+1}}{(1+\frac{\beta }{k})^{M-p+1}} \ge \frac{p+1}{p}$ for all $p=1,2,\ldots, M$. Let $N^{\star}$ be the minimum satisfying $\left(\frac{k+\beta (p+1)}{k+\beta p}\right)^{(N-1)k-1} \frac{(\frac{\beta }{k})^{M-p+1}}{(1+\frac{\beta }{k})^{M-p+1}} \ge \frac{p+1}{p}$. Note that for $\beta \ge 1$, $N^{\star} \le \theta M$, where $\theta > 1$ is a constant.
Hence we have shown that $C_n$ is a decreasing function of $k_n$ for $N\ge N^{\star}$, and therefore $k_n=1$ maximizes $C_n, \ \forall \ n=1,2,\ldots, N$. 

{\bf Case 2: Arbitrary $k_m$} 

In this case because of different scaling factor of $\frac{1}{k_m}$, the sum of the interference power $I=  \sum_{m=1, m\ne n}^N \frac{1}{k_m} \sum_{\ell=1}^{k_m} I_{\ell, m}$ is not distributed as $\chi^2$. The exact distribution of the sum of differently scaled $\chi^2$ distributed random variables is known \cite{Johnson1970}, however, is not amenable 
for analysis and does not yield simple closed form results. To facilitate analysis, we use an approximation on the sum of differently scaled $\chi^2$ distributed random variables  \cite{Feiveson1968}, which is known to be quite accurate. 

\begin{lemma}\label{lem:pdfapprox} Let $X = \sum_{i=1}^L a_i z_i$, where $a_i$'s are constants and $z_i \sim \chi^2(2)$. Then the PDF of $X$ is well approximated by the PDF of the Gamma distributed random variable with parameters $\lambda$ and $1/\alpha$, i.e. $f_X(x) = \frac{\alpha^{\lambda}}{\Gamma(\lambda)}e^{-\alpha x} x^{\lambda -1 }$, where 
$\lambda = \frac{1}{2}\frac{(\sum_{i=1}^L a_i)^2}{\sum_{i=1}^L a^2_i}$ and $\alpha = \frac{1}{2}\frac{\sum_{i=1}^L a_i}{\sum_{i=1}^L a^2_i}$.
\end{lemma}

Using Lemma \ref{lem:pdfapprox}, we can approximate the pdf of $I$ by $f_I(x) = \frac{\alpha^{\lambda}}{\Gamma(\lambda)}e^{-\alpha x} x^{\lambda -1 }$, where  $\alpha = \frac{1}{2}\frac{N\prod_{m=1,m\ne n}^Nk_m}{\sum_{m=1,m\ne n}^N k_m}$ and $\lambda  = \frac{1}{2}\frac{N^2\prod_{m=1,m\ne n}^Nk_m}{\sum_{m=1,m\ne n}^N k_m}$. With this approximation, evaluating the expectation in (\ref{eq:appexpC}) with respect to $I$, we get 
\begin{eqnarray*}
C_n 
&=&R k_n \int_{0}^{\infty}  \sum_{r=1}^{M-k_n+1}\frac{(\beta k_n x)^r}{r!} e^{-\beta k_nx} \frac{\alpha^{\lambda}x^{\lambda-1}}{\Gamma(\lambda)} e^{-\alpha x} dx, \\
&=& R k_n \alpha^{\lambda} \sum_{r=1}^{M-k_n+1}  \frac{(\beta k_n)^r}{(\alpha+\beta k_n)^{r+\lambda-1}}  \frac{(r+ \lambda-1)!}{r! (\lambda-1)!}.
\end{eqnarray*}
Using a similar argument as for the case of $k_m=k, \forall \  m \ne n$, we can show that $C_n$ is a decreasing function of $k_n$ for sufficiently large $N$. For the sake of brevity we do not repeat the argument here again.

\bibliographystyle{IEEEtran}
\bibliography{IEEEabrv,Research}

\begin{thebibliography}{10}
\providecommand{\url}[1]{#1}
\csname url@samestyle\endcsname
\providecommand{\newblock}{\relax}
\providecommand{\bibinfo}[2]{#2}
\providecommand{\BIBentrySTDinterwordspacing}{\spaceskip=0pt\relax}
\providecommand{\BIBentryALTinterwordstretchfactor}{4}
\providecommand{\BIBentryALTinterwordspacing}{\spaceskip=\fontdimen2\font plus
\BIBentryALTinterwordstretchfactor\fontdimen3\font minus
  \fontdimen4\font\relax}
\providecommand{\BIBforeignlanguage}[2]{{%
\expandafter\ifx\csname l@#1\endcsname\relax
\typeout{** WARNING: IEEEtran.bst: No hyphenation pattern has been}%
\typeout{** loaded for the language `#1'. Using the pattern for}%
\typeout{** the default language instead.}%
\else
\language=\csname l@#1\endcsname
\fi
#2}}
\providecommand{\BIBdecl}{\relax}
\BIBdecl

\bibitem{Tarokh1999a}
V.~Tarokh, H.~Jafarkhani, and A.~Calderbank, ``Space-time block coding for
  wireless communications: {Performance} results,'' \emph{{IEEE} J. Sel. Areas
  Commun.}, vol.~17, no.~3, pp. 451--460, March 1999.

\bibitem{Telatar1999}
E.~Telatar, ``Capacity of multi-antenna gaussian channels,'' \emph{European
  Trans. on Telecommunications}, vol.~10, no.~6, pp. 585--595, Nov./Dec. 1999.

\bibitem{Foschini1998}
P.~W. Wolniansky, G.~J. Foschini, G.~D. Golden, and R.~A. Valenzuela,
  ``{V-BLAST}: An architecture for realizing very high data rates over the
  rich-scattering wireless channel,'' in \emph{ISSSE-1998, Pisa, Italy}, Sept.
  1998.

\bibitem{Booktse}
D.~Tse and P.~Viswanath, \emph{Fundamentals of wireless communication}.\hskip
  1em plus 0.5em minus 0.4em\relax New York, NY, USA: Cambridge University
  Press, 2005.

\bibitem{Zheng2003}
L.~Zheng and D.~Tse, ``Diversity and multiplexing: A fundamental tradeoff in
  multiple-antenna channels,'' \emph{{IEEE} Trans. Inf. Theory}, vol.~49,
  no.~5, pp. 1073--1096, May 2003.

\bibitem{Vaze2009}
R.~Vaze and R.~Heath~Jr., ``Transmission capacity of ad-hoc networks with
  multiple antennas using transmit stream adaptation and interference
  cancelation,'' \emph{{IEEE} Trans. Inf. Theory}, submitted Dec. 2009,
  available on http://arxiv.org/abs/0912.2630.

\bibitem{Blum2003}
R.~Blum, ``{MIMO} capacity with interference,'' \emph{{IEEE} J. Sel. Areas
  Commun.}, vol.~21, no.~5, pp. 793--801, June 2003.

\bibitem{Mackay2009}
R.~Louie, M.~McKay, and I.~Collings, ``Spatial multiplexing with {MRC} and {ZF}
  receivers in ad hoc networks,'' in \emph{IEEE International Conference on
  Communications, 2009. ICC '09.}, June 2009, pp. 1--5.

\bibitem{Weber2005}
S.~Weber, X.~Yang, J.~Andrews, and G.~de~Veciana, ``Transmission capacity of
  wireless ad hoc networks with outage constraints,'' \emph{{IEEE} Trans. Inf.
  Theory}, vol.~51, no.~12, pp. 4091--4102, Dec. 2005.

\bibitem{Jindal2008a}
N.~Jindal, J.~Andrews, and S.~Weber, ``Rethinking {MIMO} for wireless networks:
  Linear throughput increases with multiple receive antennas,'' in \emph{IEEE
  International Conference on Communications, 2009. ICC '09.}, June 2009, pp.
  1--6.

\bibitem{Johnson1970}
N.~Johnson and S.~Kotz, \emph{Continuous univariate distributions}.\hskip 1em
  plus 0.5em minus 0.4em\relax New York, NY, USA: Houghton Mifflin, 1970.

\bibitem{Feiveson1968}
A.~H. Feiveson and F.~C. Delaney, ``The distribution and properties of a
  weighted sum of chi squares,'' NASA Technical Note, available on
  http://ntrs.nasa.gov/archive/nasa/casi.ntrs.nasa.gov/19680015093\_1968015093%
.pdf 1968.

\end{thebibliography}

\begin{figure}[!h]
\centering
\includegraphics[width=5in]{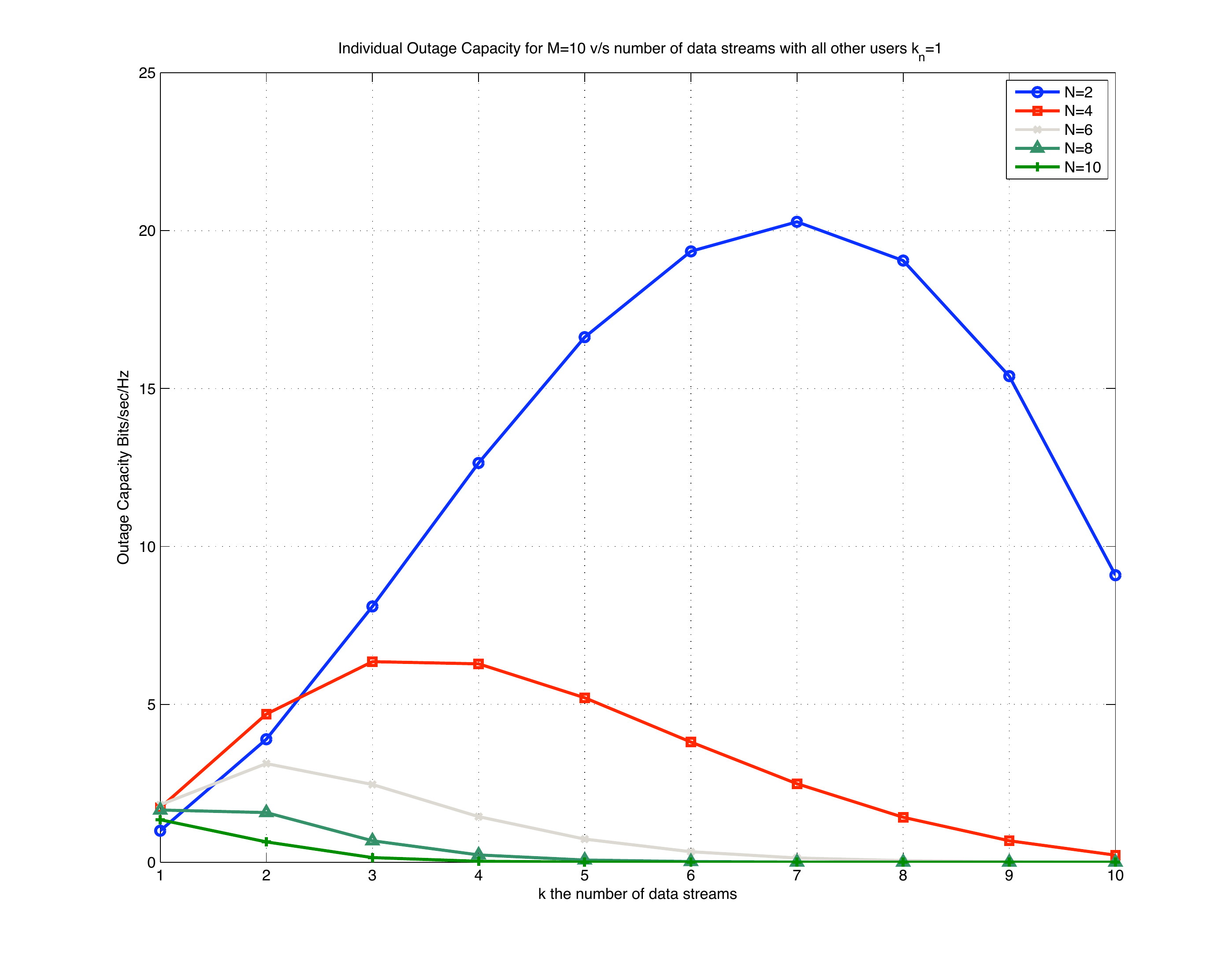}
\caption{Outage capacity of a single user (1st user) with 10 transmit antennas  (M=10) v/s number of data streams with varying  N when each of the interferers uses a single data stream  $k_n=1, \forall \ n \ne 1$.}
\label{fig:M=10}
\end{figure}

\begin{figure}
\centering
\includegraphics[width=4.5in]{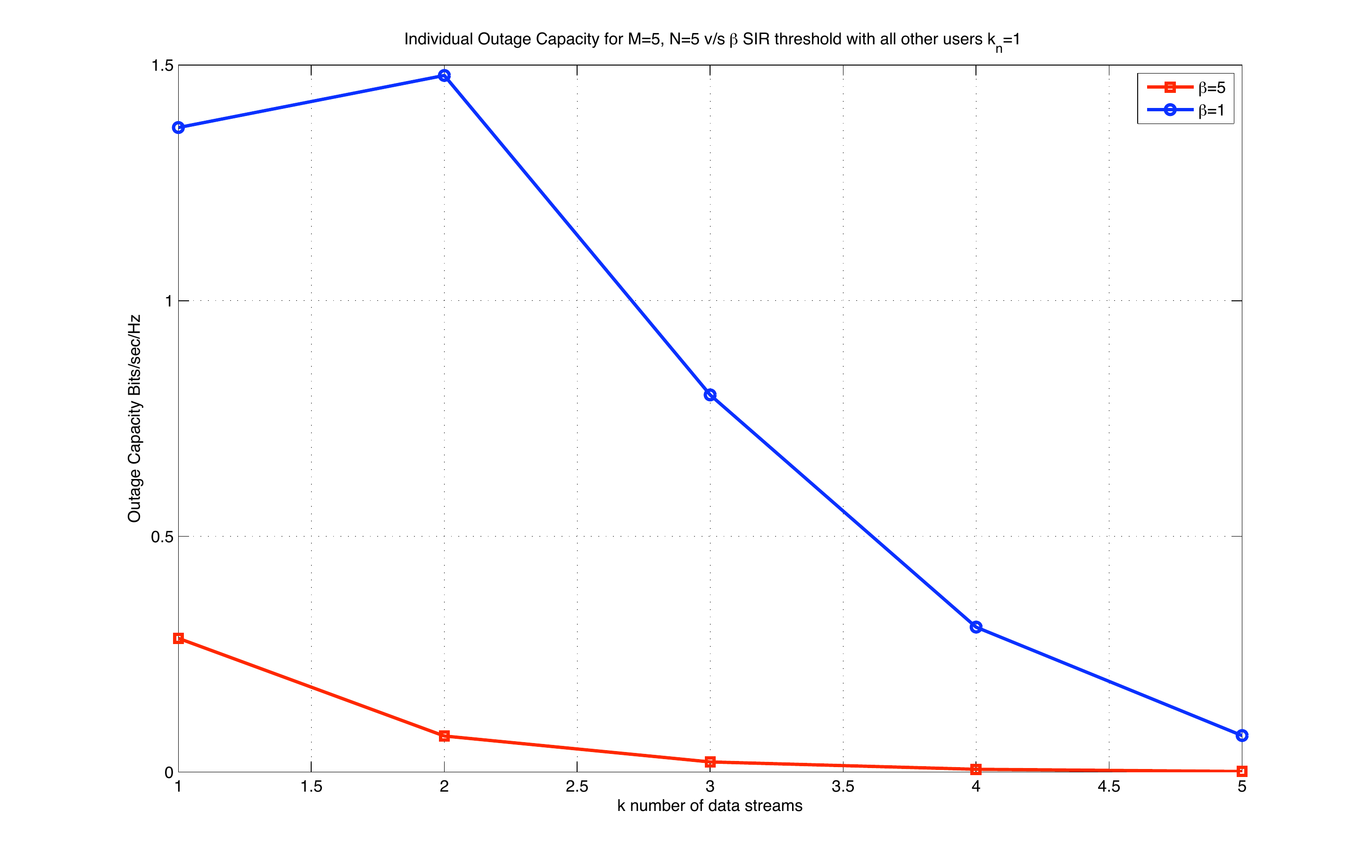}
\caption{Outage capacity of a single user (1st user) with 5 transmit antennas  (M=5) v/s number of data streams for  N=5 when each of the interferers uses a single data stream  $k_n=1, \forall \ n \ne 1$ with varying $\beta$.}
\label{fig:M=5}
\end{figure}

\begin{figure}
\centering
\includegraphics[width=4.5in]{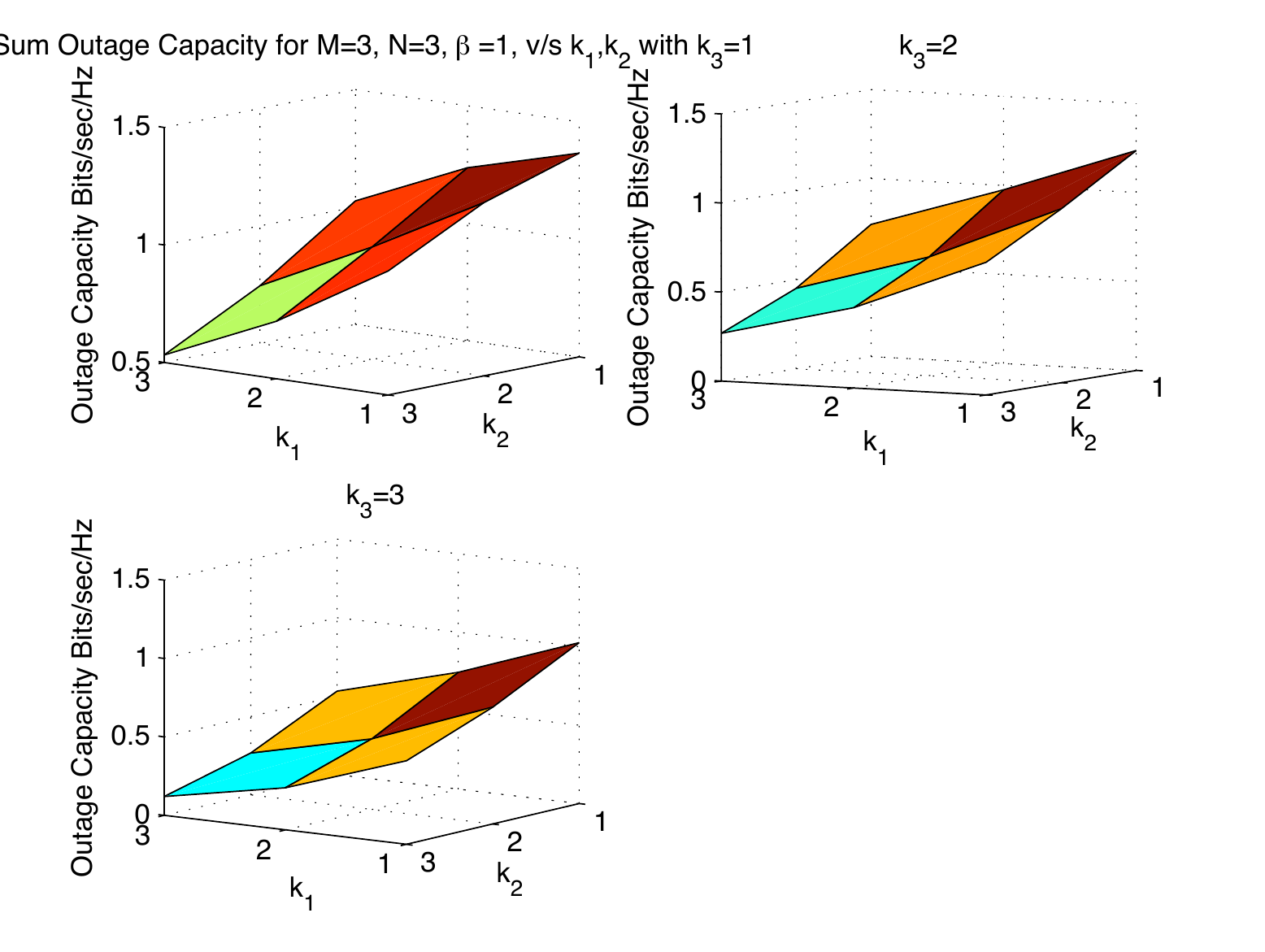}
\caption{Sum of the outage capacities for $M=3, N=3, \beta=1$ with varying $k_1,k_2,k_3$.}
\label{fig:M=3N=3}
\end{figure}

\end{document}